\documentclass[journal,9pt]{IEEEtran}
\usepackage{amsmath,amsfonts,amssymb,amsthm,mathtools} 
\usepackage{amsmath,amsthm}
\usepackage{comment}
\usepackage{euscript}
\usepackage{mathrsfs}
\usepackage{bbm}
\usepackage{flushend}
\usepackage{graphicx}
\usepackage{float}
\usepackage{subcaption}
\usepackage{color}
\usepackage{authblk}

\newtheorem{proposition}{Proposition}

\newtheorem{definition}{Definition}

\usepackage{titlesec}
\titlespacing*{\section}
{0pt}{1.5ex plus 1ex minus .2ex}{1.5ex plus .2ex}
\titlespacing*{\subsection}
{0pt}{1.5ex plus 1ex minus .2ex}{1.5ex plus .2ex}
\usepackage[linesnumbered,ruled,vlined]{algorithm2e}
\SetKwInput{KwInput}{Input}                
\SetKwInput{KwOutput}{Output}              

\title{K-sample Multiple Hypothesis Testing for Signal Detection}

\author[1]{Uriel Shiterburd}
\author[1]{Tamir Bendory}
\author[2]{Amichai Painsky}

\affil[1]{School of Electrical Engineering,	Tel Aviv University, Tel Aviv, Israel}
\affil[2]{Department of Indusrial Engineering, Tel Aviv University, Tel Aviv, Israel}

\begin{document}
	
	\maketitle
	
	\begin{abstract}
		This paper studies the classical problem of estimating the locations of signal occurrences in a noisy measurement.
	    Based on a  multiple hypothesis testing  scheme,
		we design a $K$-sample statistical test to control the false discovery rate (FDR). 
		Specifically, we first convolve the noisy measurement with a smoothing kernel, and find all local maxima. 
        Then, we	evaluate the joint probability of $K$ entries in the vicinity of each local maximum, derive the corresponding p-value, and  apply the Benjamini-Hochberg procedure to account for multiplicity.
		We demonstrate through extensive experiments that our proposed method, with $K=2$,  controls the prescribed FDR while increasing the power compared to a one-sample test.
	\end{abstract}
	%
	%
	\section{Introduction}
	\label{sec:intro}
	
	We study the  classical problem of detecting the locations of multiple signal occurrences in a noisy measurement. 
	Figure~\ref{fig:into-ex} illustrates an instance of the problem,  where the goal is to estimate the locations of  five signal occurrences from the noisy observation.
	This detection problem is an essential step in a wide variety of  applications, including single-particle cryo-electron microscopy~\cite{frank2006three,scheres2012relion}, neuroimaging brain mapping~\cite{taylor2007detecting}, microscopy \cite{egner2007fluorescence,geisler2007resolution}, astronomy \cite{brutti2005spike} and retina study~\cite{baccus2002fast}.

	Assuming the shape of the signal is approximately known, the most popular heuristic to detect signal occurrences is  correlating the noisy measurement with  the signal, and choosing the peaks of the correlation as the estimates of the signals' locations. 
	If the signal is known, then this method is referred to as 
	\textit{a matched filter} \cite{pratt2007digital}.
	However, extending  the matched filter to multiple signal occurrences is not trivial, and may lead to  high false-alarm rates. This mostly manifests in high noise levels, and in cases where the signal occurrences are densely packed~\cite{Vio_2019}.
	
	Following the work of Schwartzman et al.~\cite{schwartzman2011multiple,cheng2017multiple},	we study the signal detection problem from the perspective of multiple hypothesis testing.
	In particular, Schwartzman et al.~\cite{schwartzman2011multiple} proposed 
	to choose the local maxima of the measurement, after it was correlated with a template which ideally resembles the underlying signal, as the estimators of the locations of the signal occurrences. 
	While this procedure is akin to matched filter, the main contribution of~\cite{schwartzman2011multiple} is setting the threshold of the detection problem based on the Benjamini-Hochberg procedure~\cite{benjamini1995controlling}  to control the false discovery rate  under a desired level.
	In Section \ref{sec:background} we introduce the results of \cite{schwartzman2011multiple} in detail. 
	Notably, the statistical test of  Schwartzman et al. exploits only  local maxima, while ignoring their local neighborhoods. Hereafter, we refer to this scheme as a one-sample statistical test.

	In this work, we extend the one-sample statistical test of~\cite{schwartzman2011multiple}  to account for the neighborhood of the local maxima. 
	In particular, we claim that taking the  neighborhood of the local maxima, or its shape, may strengthen the statistical test. For example, if the noise is highly correlated, a narrow and sharp peak is more likely to indicate on a signal occurrence than a wider peak.
	The first contribution of this paper is extending the 	framework of~\cite{schwartzman2011multiple} to a statistical test which is based on multiple samples around each local maxima; we refer to this procedure as the $K$-sample statistical test. 
	This entails calculating the joint probability distribution of the chosen $K$ samples and	deriving the corresponding p-value required for the Benjamini-Hochberg procedure. 
	
	The second contribution of this work is implementing the new framework for two-sample test ($K=2$), namely, when the statistical test relies on the local maxima and one additional point in their vicinity. 
	As discussed in Section \ref{sec:main},  the computational complexity of the K-sample test grows exponentially fast with $K$, and thus we leave the implementation of the $K$-sample test with $K>2$ for future work. We conduct comprehensive numerical experiments, and show that the two-sample test outperforms the one-sample test of~\cite{schwartzman2011multiple} in a wide range of scenarios.
	Figure~\ref{fig:into-ex} illustrates a simple example when our two-sample test outperforms the one-sample test of~\cite{schwartzman2011multiple}. 

	\begin{figure}[h]
		\centering
		\includegraphics[width=1\linewidth]{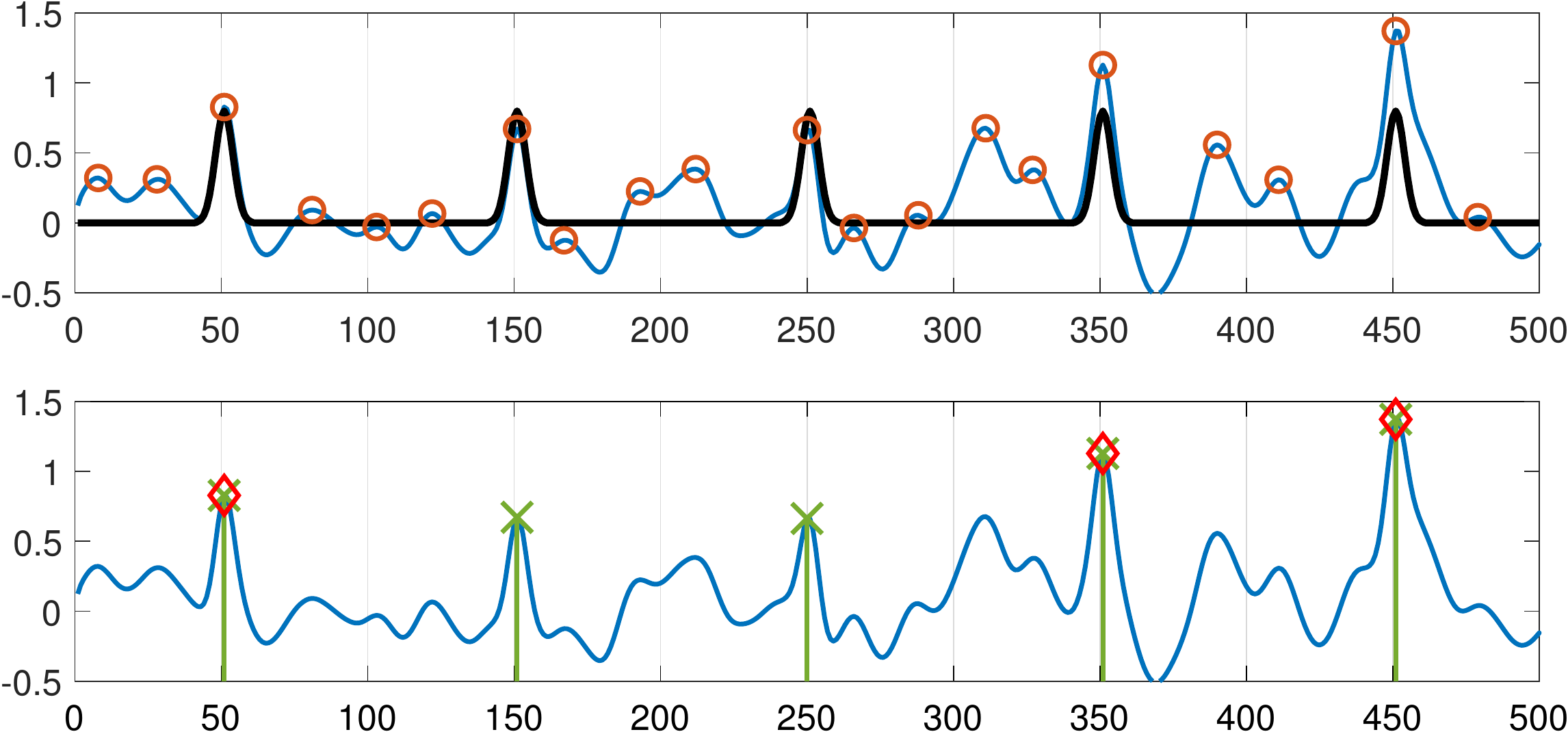}
		\caption{ The upper panel shows a noiseless observation with five signal occurrences (black), and a noisy measurement (blue) with the local maxima marked by orange circles. Our goal is to detect the locations of the five true signal occurrences from the noisy measurement.
			The bottom panel shows in red the result of the one-sample test of~\cite{schwartzman2011multiple} that captures only 3 signal occurrences, while our algorithm (Algorithm~\ref{alg:two sample}) accurately identifies all five signal occurrences (in green crosses). 
			\label{fig:into-ex}}
		\vskip -5pt
	\end{figure}

	The rest of the paper is organized as follows.  Section~\ref{sec:background} provides mathematical and statistical background. 
	Section~\ref{sec:main} presents the main result of this paper. Section~\ref{sec:exp} discusses computational  considerations and presents  numerical results. Section~\ref{sec:conclusion} concludes the paper and delineates future research directions.

	\section{Mathematical model and background}	\label{sec:background}

	Let $h_j(t) \in C(\mathbb{R})$ be a continuous signal  with a compact support~$S_j$ of unknown location and $\int_{s\in S_j}h_j(s) = 1$. 
	Let 
	\begin{equation}
		\mu(t) = \sum_{j = -\infty}^{\infty} a_j h_j(t), 
	\end{equation}
	be a train of unknown signals, where $a_j >0$.
	Our goal is to estimate the locations of the signals $h_j$ from a noisy measurement $y(t) \in C(\mathbb{R})$ of the form
	\begin{equation}\label{singalpnoise}
		y(t) = \mu(t) + z(t), 
	\end{equation}
	where $z(t)$ denotes the noise.

	In our problem, the null hypothesis is that there are no signals, namely, a pure noise measurement $y(t)=z(t)$. This problem can be thought of as a special case of multiple hypothesis testing \cite{rupert2012simultaneous}, where our hypothesizes are the points at which we suspect the signal occurrences are located.
	
	\subsection{Multiple hypothesis testing}
	
	
		In a classical (single) hypothesis testing problem, one considers two hypotheses (the null and the alternative) and decides  whether the data at hand sufficiently support the null. 
	    Specifically, given a set of observations (e.g., coin tosses), our goal is to increase the probability of making a true discovery (e.g., declare that a coin is biased given that it is truly biased), while maintaining a prescribed level of type I error, $\alpha$ (declaring that a coin is biased given that it is not). This is typically done by examining the p-value: the probability of obtaining test results at least as extreme as the result actually observed, under the assumption that the null hypothesis is correct. 
		Specifically, we declare a discovery if the p-value is not greater than $\alpha$. 
		
	Multiple testing arises when a statistical analysis involves multiple simultaneous tests, each of which has a potential to produce a ''discovery". A stated confidence level generally applies only to each test considered individually, but often it is desirable to have a confidence level for the whole family of simultaneous tests. Failure to compensate for multiple comparisons may produce severe consequences. For example, consider a set of $m$ coins, where our goal is to detect biased coins. Assume we test the hypothesis that a coin is biased with a confidence level of $\alpha$. Further, assume that all the coins are in-fact unbiased. Then, we expect to make $\alpha m$ false discoveries. Further, and perhaps more importantly, the probability that we make at least one false discovery (denoted \textit{family wise error rate}, FWER) increases as $m$ grows. This means we are very likely to report that there is at least one biased coin in hand, even if they are all fair.   
	A classical approach for the multiple testing problem is to control the FWER. This is done by applying the Bonferroni correction:  it can be shown that by performing each test in a confidence level of $\alpha/m$, we are guaranteed to have a FWER not greater than $\alpha$. The Bonferroni correction is very intuitive and simple to apply. Yet, it is also very conservative and thus very few discoveries are likely to be made.


	To overcome these caveats, Benjamini and Hochberg \cite{benjamini1995controlling} proposed an alternative criterion. Instead of controlling the FWER, they suggested controlling the \textit{false discovery rate} (FDR). The FDR is formally defined as the expected proportion of false discoveries from all the discoveries that are made. Thus, we allow more false discoveries if we make more discoveries. 
\begin{definition}
	The false discovery rate (FDR) is defined by 
	\begin{equation}
		\text{FDR} = \mathbb{E}\left( V/R|R>0\right),
	\end{equation}
	where $V$ and $R$ are, respectively, the number of false positives and all positives (rejections of null hypothesis).
\end{definition}
This criterion is evidently less conservative than the FWER. Thus, FDR-controlling procedures have greater power, at the cost of increased numbers of type I errors. Benjamini and Hochberg introduced a simple approach for controlling the FDR~\cite{benjamini1995controlling}. Their approach is based on sorting the p-values in ascending order $p_{(1)}\leq \dots\leq  p_{(m)}$. For a desired level $\alpha$, the Benjamini–Hochberg  procedure suggests finding  the largest $k$ such that $p_{(k)}\le \frac{k}{m}\alpha$ and reject null hypothesis for all observations of equal or less p-value.
The Benjamini-Hochberg procedure made a tremendous impact on modern statistical inference~\cite{benjamini2010simultaneous}, with many important applications to biology \cite{mcdonald2009handbook}, genetics~\cite{love2014moderated,tusher2001significance}, bioinformatics \cite{huang2009bioinformatics}, statistical learning~\cite{efron2004least}, to name but a few.

	\subsection{Controlling the FDR of local maxima}
	Going back to our problem, we aim to control the FDR of estimating the locations of the signal occurrences from the measurement $y(t)$.
	A direct approach is   testing on every sampled point in the measurement~\cite{taylor2007detecting}.
	However, this approach overlooks the spatial structure of the signals, and correlations in the measurement. In~\cite{schwartzman2011multiple},   Schwartzman, Gavrilov and Adler developed  a method to control the FDR, while  testing only  the local maxima  of a smoothed version of the measurement. 
	Specifically, 
	the algorithm of~\cite{schwartzman2011multiple} 
	begins with correlating the measurement $y(t)$  with a smoothing  kernel $w_\gamma(t)$ of compact connected support satisfying $\int_{t\in \mathbb{R}} w_\gamma(t) = 1$. This results in 
	\begin{equation}\label{eq: smoothing a sig}
		y_\gamma(t)= w_\gamma(t)*y(t) = \mu_\gamma(t) + z_\gamma(t),
	\end{equation}
	where $\mu_\gamma(t)$ and $z_\gamma(t)$ denote the smoothed signal train and noise, respectively.
	The next step is finding the set of all local maxima of the smoothed measurement:
	\begin{equation} \label{eq:T}
		T = \left\{t\in \mathbb{R}\,| \,\dot{y}_\gamma(t) = 0 \,,\, \ddot{y}_\gamma(t) < 0\right\}.
	\end{equation}
	To apply the Benjamini-Hochberg procedure, we next need to compute the p-value of  each local maxima $t\in T$ with respect to  the  null hypothesis ${\mathcal{H}_0 : y_\gamma(t) = z_\gamma(t)}$. 
	The p-value $p_\gamma(t)$ is given by
	\begin{equation} \label{eq:p_gamma}
		F_\gamma(y_\gamma(t)) = \text{Prob}(z_\gamma(t)> y_\gamma(t)\,|\, t \in {T}).
	\end{equation}
	This distribution, describing the probability of local maxima of $z_\gamma(t)$, is called the Palm distribution~\cite[Chapter 11]{cramer2013stationary}, and is given by the following CDF:
	\begin{equation}\label{eq:pval 1d}
		\begin{aligned}
			F_\gamma(u) = 1- \Phi\left(u\sqrt{\frac{\lambda_4}{\Delta}} \right)   +\sqrt{\frac{2\pi \lambda_2^2}{\lambda_4\sigma_\gamma^2}} 
			\phi\left(\frac{u}{\sigma_\gamma} \right)\Phi \left(u\sqrt{\frac{\lambda_2^2}{\Delta\sigma_\gamma^2}} \right),
		\end{aligned}
	\end{equation}
	where $\phi$ and $\Phi$ are the  normal density and its CDF respectively, while $\Delta = \sigma_\gamma^2\lambda_4-\lambda_2^2$, and 
	$\sigma_\gamma^2 = \text{Var}\{z_\gamma(t)\},\, \lambda_2 = \text{Var}\{\dot{z}_\gamma(t)\} , \,   \lambda_4 = \text{Var}\{\ddot{z}_\gamma(t)\},$	
	are the first three  moments of $z_\gamma(t)$.
	
	Finally, we apply  the Benjamini–Hochberg procedure, to control the  FDR at a desired level $\alpha$. The algorithm of Schwartzman et al., dubbed one-sample test, is summarized in Algorithm~\ref{alg:one sample}.
	Interestingly, it was shown  that this procedure achieves maximal power under asymptotic conditions~\cite{schwartzman2011multiple}.

	\begin{algorithm}[h]
		\DontPrintSemicolon
		\caption{One-sample test~\cite{schwartzman2011multiple}}\label{alg:one sample}
		\KwInput{Measurement $y(t)$, desired FDR level $\alpha$, and a smoothing kernel $w_\gamma(t)$}
		\KwOutput{A set $M$ of detected signal occurrences}
		Calculate $y_\gamma(t)$ according to (\ref{eq: smoothing a sig}) \;
		Find a set $T$~\eqref{eq:T} of all local maxima of $y_\gamma(t)$\;
		Calculate  the p-value of each point in $T$ according to~\eqref{eq:p_gamma}\;
		Get a set of estimated locations of signal occurrences $M \subset T$ using the Benjamini-Hochberg procedure for a desired FDR level $\alpha$
	\end{algorithm}

	\section{Main Result}\label{sec:main}
	Careful examination of local maxima reveals that not all peaks of the  same height are also of  a similar shape: the  existence of a signal occurrence underneath the noise varies the general shape of the local maximum; see Figure~\ref{fig:signal under noise} for an example.
	Accordingly, we suggest taking into account the vicinity  of the local maxima in the estimation process.
	In this section, we provide an explicit expression of the joint distribution of the local maximum and its closely-located samples. 
	This in turn allows us to  calculate the corresponding p-value which can be used to control the FDR using the Benjamini-Hochberg procedure. 
	
	\begin{figure}[h] 
		\centering
		\includegraphics[width=0.7\linewidth]{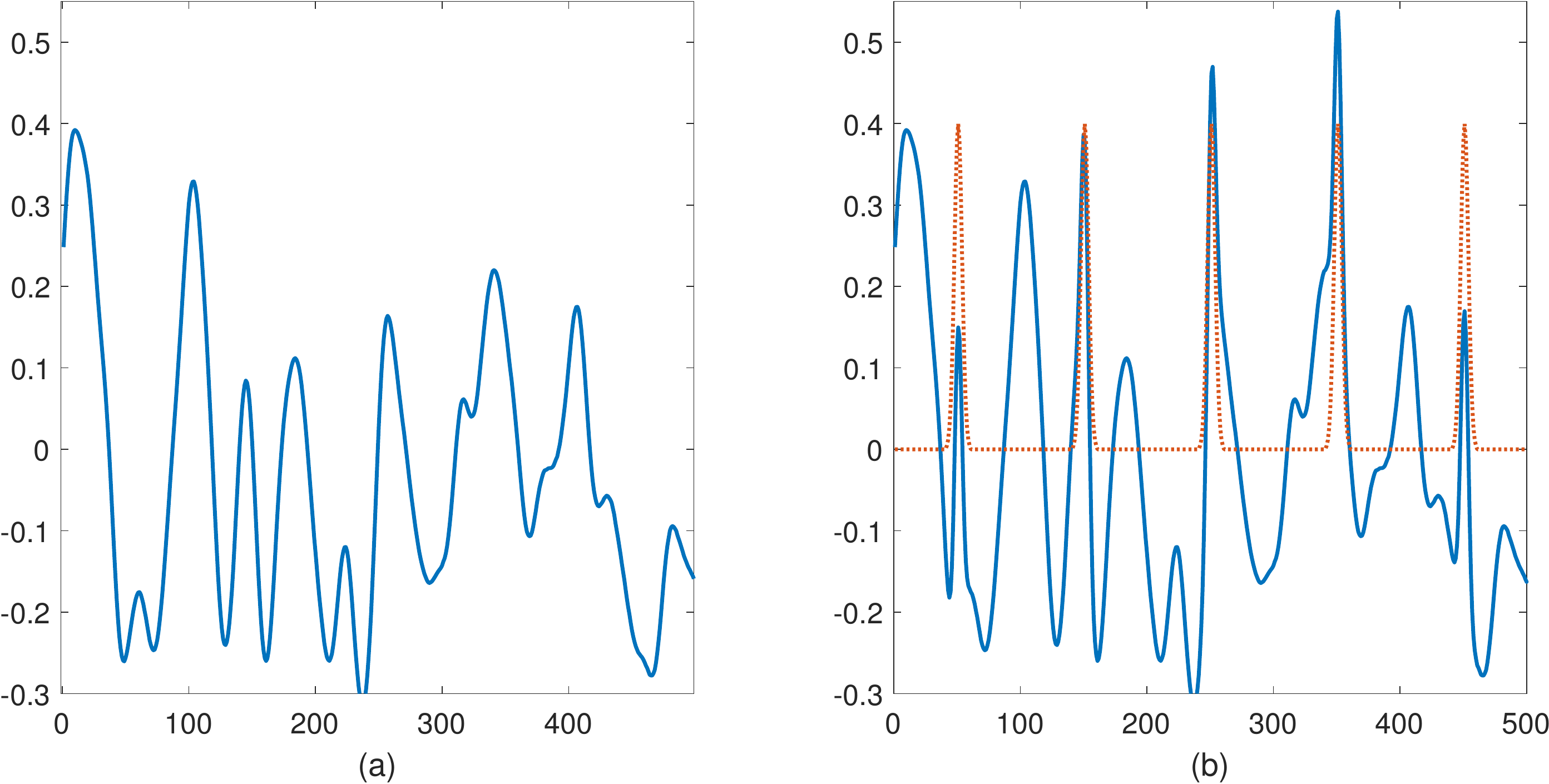}
		\caption{(a) Pure noise measurement. (b) The measurement from (a) with additional five signal occurrences (in red). As can be seen, the measurement in the vicinity of the signal occurrences is very peaky, although the corresponding local maxima are not necessarily high. 
			This motivates us to consider the vicinity of each local maxima, and not merely its height. 	
		}
		\label{fig:signal under noise}
		\vskip -5pt
	\end{figure}

	Let $z_{m}$ be the  maximum of a random   Gaussian vector $\underline{z}$. We assume that $\underline{z}$ consists of no more then a single local maximum. 
	Our goal is to evaluate  the joint distribution of $z_m$ and its $K-1$ neighbors, ordered in an arbitrary manner $\{z_{(i)}\}^{K-1}_{i=1}$,
	\begin{equation} \label{eq:joint_distribution}
		P(z_m, z_{(1)},z_{(2)}, \cdots,z_{(K-1)}).
	\end{equation}
	Applying the chain rule, we have
	\begin{align}\label{eq:generalCase}
		&P(z_m, z_{(1)},z_{(2)}, \ldots,z_{(K-1)})\nonumber\\ &= 
		P(z_{(1)},z_{(2)}, \ldots,z_{(K-1)}|z_m)P(z_m)\\
		&=P(z_{(K-1)}| z_{(K-2)}, \ldots,z_{(1)},z_m) \nonumber\\
		&\times P(z_{(K-2)}| z_{(K-3)}, \ldots,z_{(1)},z_m)\times \cdots \times P(z_{(1)}|z_m)P(z_m). \nonumber
	\end{align}
	The one-dimensional distribution of a local maxima $P(z_{m})$, or Palm distribution as mentioned earlier, was derived in~\cite{cramer2013stationary}. 
	Therefore, the following result  allows us to evaluate the joint distribution~\eqref{eq:joint_distribution}.
	
	\begin{proposition}
		Let $z_m=\text{max}({\underline{z}})$ be the maximum of a multivariate Gaussian vector $\underline{z}\in \mathbb{R}^{K}$. Let $\{z_{(i)}\}^{K-1}_{i=1}$ be an arbitrary order of points other than $z_m$. Define
		\begin{align}\nonumber
			\Phi_1(&\underline{z},i)=\\\nonumber
			&\int_{-\infty}^{z_m}\cdots 	\int_{-\infty}^{z_m}P(z_{(i+1)},\cdots,z_{(K-1)},z_{(i)}|z_{(i-1)},\cdots,z_{(1)},z_{m})\\\nonumber
			&\quad\quad \quad\quad \quad\quad\;\; dz_{(i+1)}\dots dz_{(K-1)},
		\end{align}
	and 	
		\begin{align}\nonumber
			\Phi_2(\underline{z},i&)=F(z_{(i+1)} = z_m,\cdots,z_{(K-1)} = z_m|z_{(i-1)},\cdots,z_{(1)},z_{m}).
		\end{align}
		Then, for any $i \le K-1$, we have:
		\begin{align}\nonumber
			P(z_{(i)}|z_{(i-1)},\cdots,&z_{(1)},z_m)=\frac{\Phi_1(\underline{z},i)}{
				\Phi_2(\underline{z},i)} 
			\mathbbm{1} (\cap_{j=1}^{i-1}\{z_{(j)}\le z_m\}),
		\end{align}
		where $F(\cdot)$ is the cumulative  distribution function (CDF) of $P(\cdot)$ and $\mathbbm{1}(\cdot)$ is an indicator function.
	\end{proposition}
	\begin{proof}
		Using  Bayes' Theorem, we have:
		\begin{align}
			&P\left(z_{(i)}\,|\,z_{(i-1)},\cdots,z_{(1)},z_m=\max\{\underline{z}\}\right)=\nonumber \\
			&P\left(z_{(i)}\,|\,z_{(i-1)},\cdots,z_{(1)},z_m, \{z_{(j)}\le z_m\}_{j=1}^{K-1}\right)=\\
			&\frac{P(z_{(i)},\{z_{(j)}\le z_m\}_{j=1}^{K-1}\,|\,z_{(i-1)},\cdots,z_{(1)},z_m)}{P(\{z_{(j)}\le z_m\}_{j=1}^{K-1}\,|\,z_{(i-1)},\cdots,z_{(1)},z_m)}. \label{basyes 1st} \nonumber
		\end{align}
		The denominator satisfies
		\begin{equation}
			\begin{aligned}\nonumber
				P(\{z_{(j)}&\le z_m\}_{j=1}^{K-1}\,|\,z_{(i-1)},\cdots,z_{(1)},z_m)=\\
				&F(z_{(i+1)} = z_m,\cdots,z_{(K-1)} = z_m|z_{(i-1)},\cdots,z_{(1)},z_{m})\cdot\\
				&\mathbbm{1} (z_{(i-1)}\le z_m , \cdots, z_{(1)}\le z_m),
			\end{aligned}
		\end{equation}
		while the numerator
		\begin{align}\nonumber
			&P(z_{(i)},\{z_{(j)}\le z_m\}_{j=1}^{K-1}\,|\,z_{(i-1)},\cdots,z_{(1)},z_m)=\\\nonumber
			&\int_{-\infty}^{z_m}\cdots 	\int_{-\infty}^{z_m}P(z_{(i+1)},\cdots,z_{(K-1)},z_{(i)}|z_{(i-1)},\cdots,z_{(1)},z_{m})\cdot \\\nonumber 
			&\;\;\quad\quad\quad\quad \quad\quad dz_{(i+1)} \cdots dz_{(K-1)}\mathbbm{1} (z_{(i-1)}\le z_m , \cdots, z_{(1)}\le z_m).
		\end{align}
	\end{proof}

	Equipped with the explicit expression of the joint distribution of the local maximum and its neighbors, we are ready to define the corresponding p-value. 
	The p-value is essential in order to apply the Benjamini-Hochberg procedure. 
	In this work the p-value at $\underline{s}=[s_1,\ldots s_n]$ is defined by: 
	\begin{equation} \label{eq:p_val}
		\text{p-value}(\underline{s}) = \text{Prob}(z_1>s_1 , \cdots, z_n>s_n ).
	\end{equation}
	We emphasize that alternative definitions of p-value may be used, and we chose the conservative~\eqref{eq:p_val} to mitigate the computational complexity, which is the focus of the next section. 
	Our proposed scheme is summarized in Algorithm~\ref{alg:two sample}.

	\section{Numerical experiments}\label{sec:exp}
	
	Computing the p-value corresponding to~\eqref{eq:joint_distribution} is computationally demanding as one needs to evaluate  multi-dimensional integrals. 
	In particular, the computational burden  grows quickly with $K$: the number of samples we consider around each local maximum. 
	Therefore, we implemented our algorithm for a single neighbor (namely, $K=2$) at a distance $d$ from the local maxima; we refer to this algorithm as a two-sample test. 	
	In addition, to mitigate the computational burden, we approximate the marginal distribution as  a 
	truncated Gaussian distribution. 
	This is only an approximation that we used to alleviate the computational load since the marginal of a truncated Gaussian distribution is not necessarily a truncated Gaussian distribution. 
	Nevertheless, we found empirically that this approximation provides good numerical results.
	
	\begin{algorithm}[h]
		\DontPrintSemicolon
		\caption{$K$-sample test}\label{alg:two sample}
		\KwInput{Measurement $y(t)$, desired FDR level $\alpha$,  a smoothing kernel $w_\gamma(t)$, and number of  neighbors $K-1$}
		\KwOutput{A set $M$ of detected signal occurrences}
		Calculate $y_\gamma(t)$ according to (\ref{eq: smoothing a sig}) \;
		Find a set $T$~\eqref{eq:T} of all local maxima of $y_\gamma(t)$\;
		Calculate  the p-value for each point  in $T$ and its $(K-1)$ neighbors according to~\eqref{eq:p_val}\;
		Get a set of estimated locations of signal occurrences $M \subset T$ using the Benjamini-Hochberg procedure for a desired FDR level $\alpha$

	\end{algorithm}

	We compared Algorithm~\ref{alg:two sample} with $K=2$ with the one-sample test of~\cite{schwartzman2011multiple}, Algorithm~\ref{alg:one sample}. 
	We set the length of the measurement to $L=1000$, with 10 
	truncated Gaussian signals 
	\begin{equation} \label{eq:signal}
		ah_j(t) = \frac{a}{b}\phi\left( \frac{t-\tau_j}{b}\right) \mathbbm{1}[-cb,cb],
	\end{equation}
	with $b = c = 3$. 
	The parameter $b$, the standard deviation of the Gaussian, controls the width of the signal, and has a significant influence on the performance of the algorithms. 
	Following~\cite{schwartzman2011multiple}, we generated the noise $z(t)$ with parameter $\nu >0$ according to $z(t) =	\sigma \int_{-\infty}^{\infty} \frac{1}{\nu} \phi\left(\frac{t-s}{\nu}\right)dB(s),$
	where $B(s)$ is Brownian motion and $\sigma$ is a noise parameter. Notice that this noise $z(t)$ is a convolution of white Gaussian noise of zero mean and variance $\sigma^2$ with a CDF of Gaussian with  zero mean and variance $\nu$. The latter can be seen as a parameter governing the noise correlation level.
	This specific noise distribution was chosen because of available good approximations of its three first moments:
	$\sigma_\gamma^2 = \frac{\sigma_0^2}{2\sqrt{\pi}\nu}$ , $\lambda_2 = \frac{\sigma_0^2}{4\sqrt{\pi}\nu^3}$ 
	and $\lambda_4 = \frac{3\sigma_0^2}{9\sqrt{\pi}\nu^5}$. Those approximations will be useful in calculating p-values according to~\eqref{eq:pval 1d}.
	The smoothing kernel is chosen to be $w_\gamma(t) = \frac{1}{\gamma}\phi(t/\gamma)$ with varying values of $\gamma$.
	We set the  FDR level to be $\alpha=0.05$.
	
	Following~\cite{schwartzman2011multiple}, we define a true positive to be any estimator inside the support of a true signal occurrence. 
	Similarly, a false positive is an estimator  outside the support of any true signal occurrence. 
	We compare the algorithms 
	for different values of $\nu$ (the variance of the Gaussian used to generate a correlated noise),  $b$ (the variance of  the Gaussian signal), and $\gamma$ (the variance of the Gaussian smoothing kernel). For each set of parameters, we present the average power and  false discovery rate, averaged over 1000 trials.

	\begin{figure}
		\begin{subfigure}[ht]{0.15\textwidth}
		\centering
		\includegraphics[width=\columnwidth]{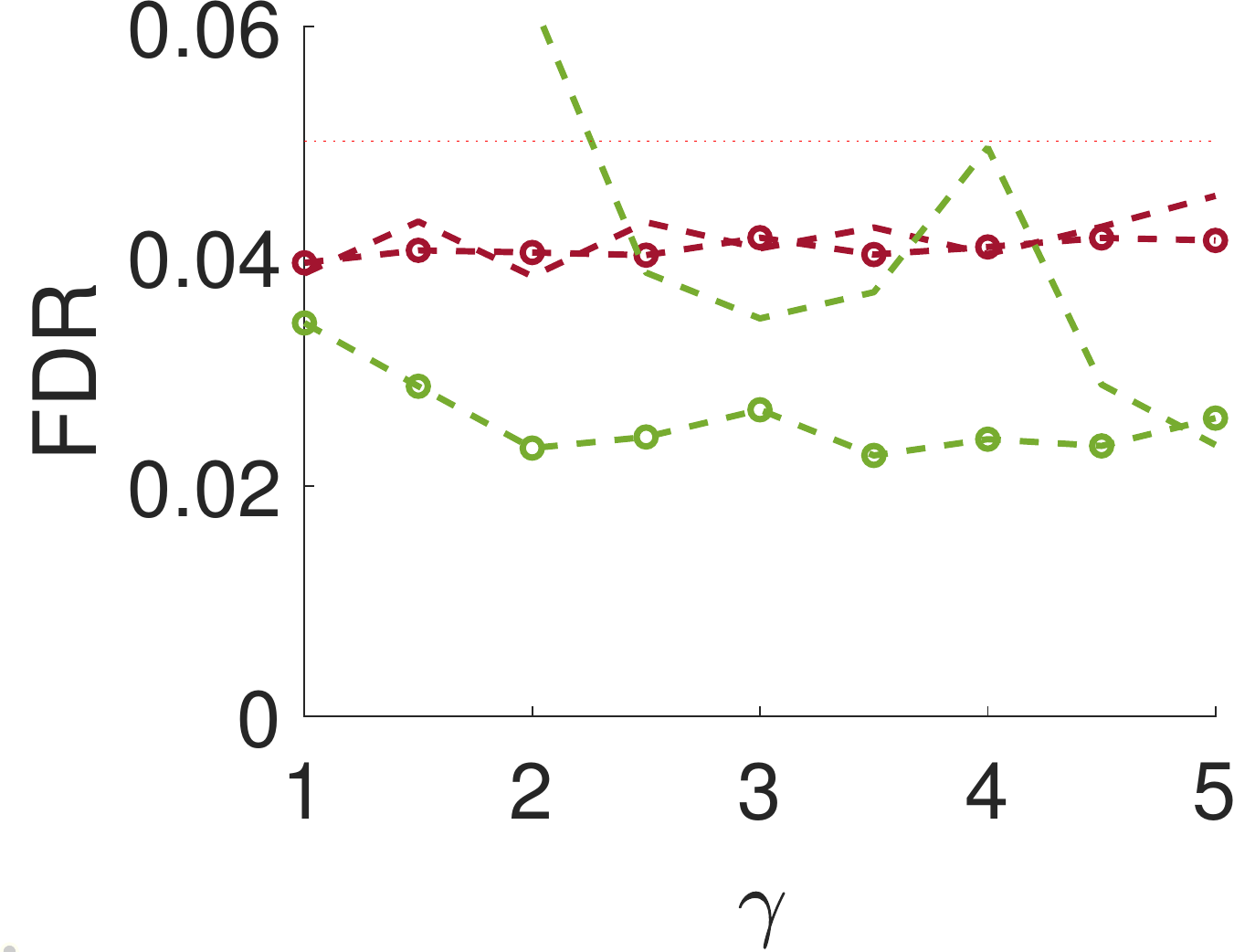}
		\caption{\label{fig:neig_2_FA_nu_3} $\nu = 3$}
	\end{subfigure}
	\hfill
	\begin{subfigure}[ht]{0.15\textwidth}
	\centering
	\includegraphics[width=\columnwidth]{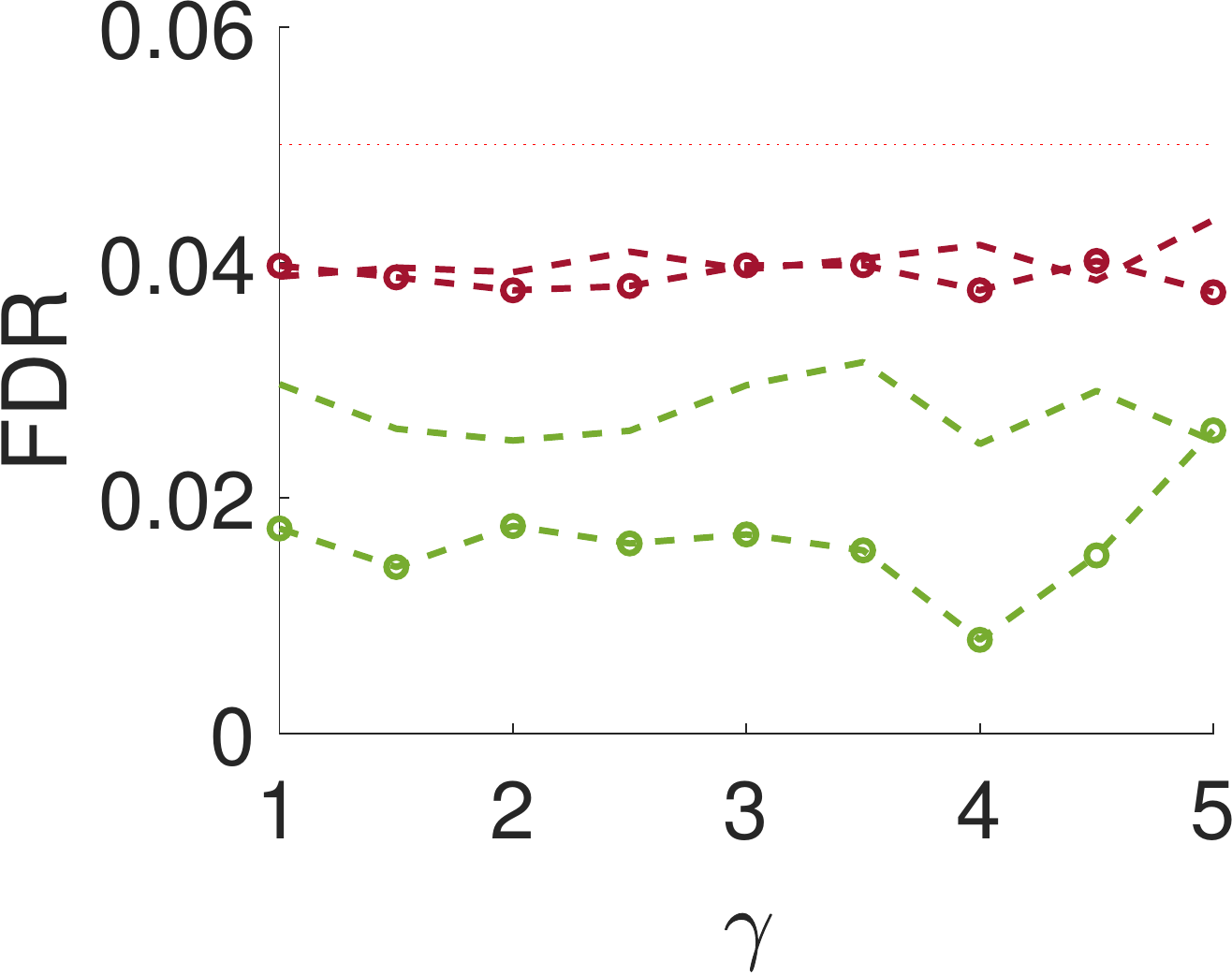}
	\caption{\label{fig:neig_2_FA_nu_4} $\nu = 4$}
\end{subfigure}
\hfill
\begin{subfigure}[ht]{0.15\textwidth}
\centering
\includegraphics[width=\columnwidth]{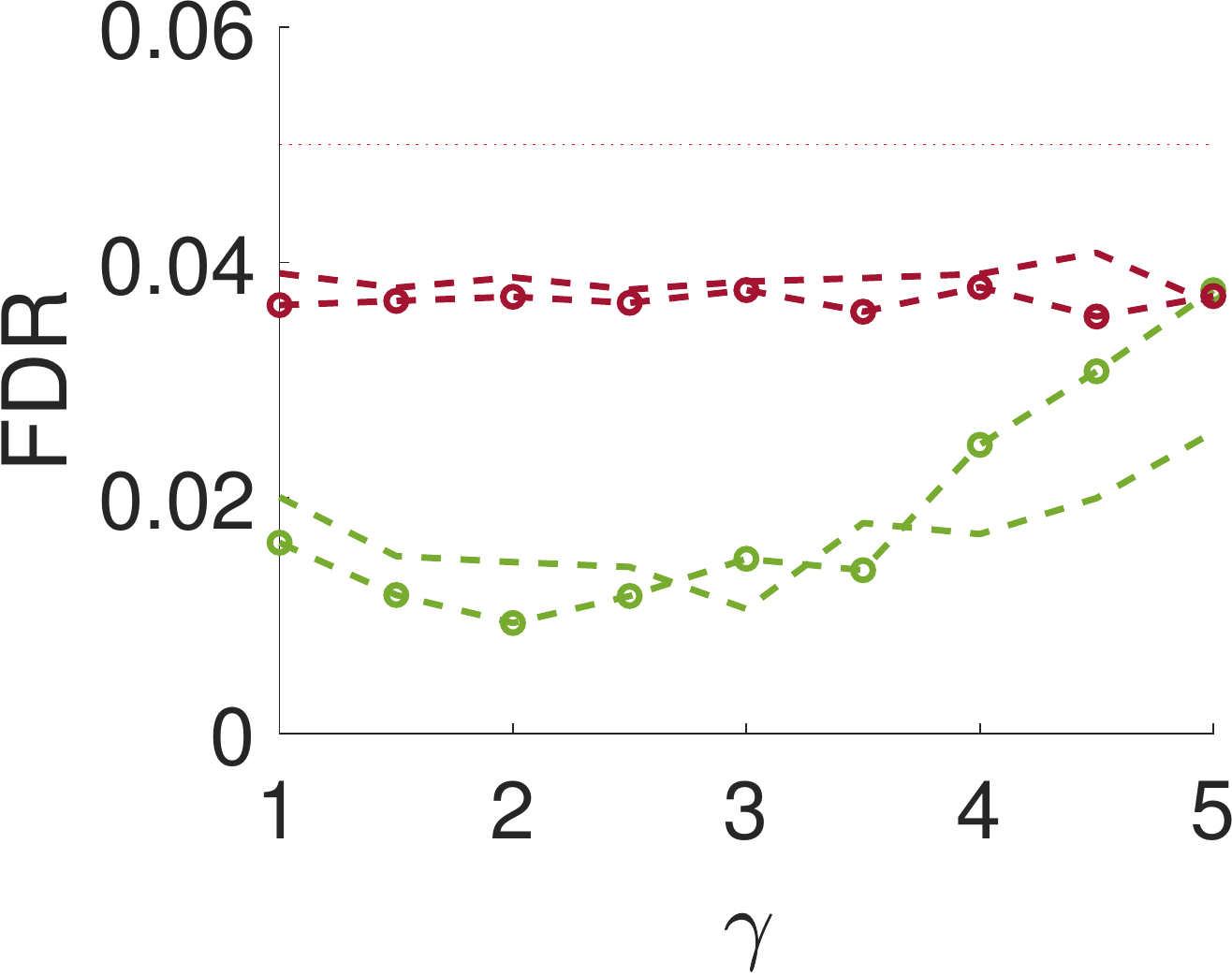}
\caption{\label{fig:neig_2_FA_nu_5} $\nu = 5$}
\end{subfigure}

\begin{subfigure}[ht]{0.15\textwidth}
\centering
\includegraphics[width=\columnwidth]{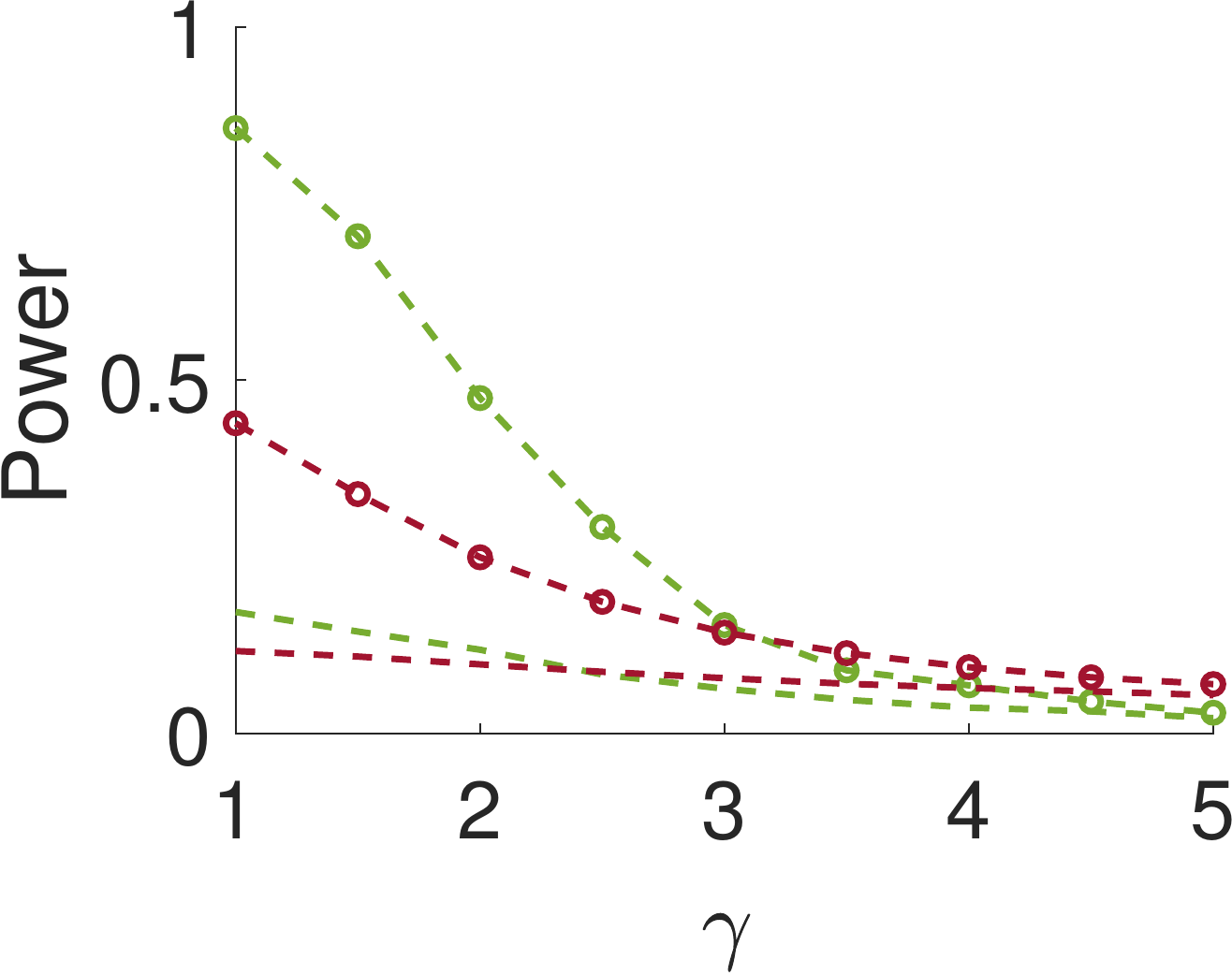}
\caption{\label{fig:neig_2_Power_nu_3} $\nu = 3$}
\end{subfigure}
\hfill
\begin{subfigure}[ht]{0.15\textwidth}
\centering
\includegraphics[width=\columnwidth]{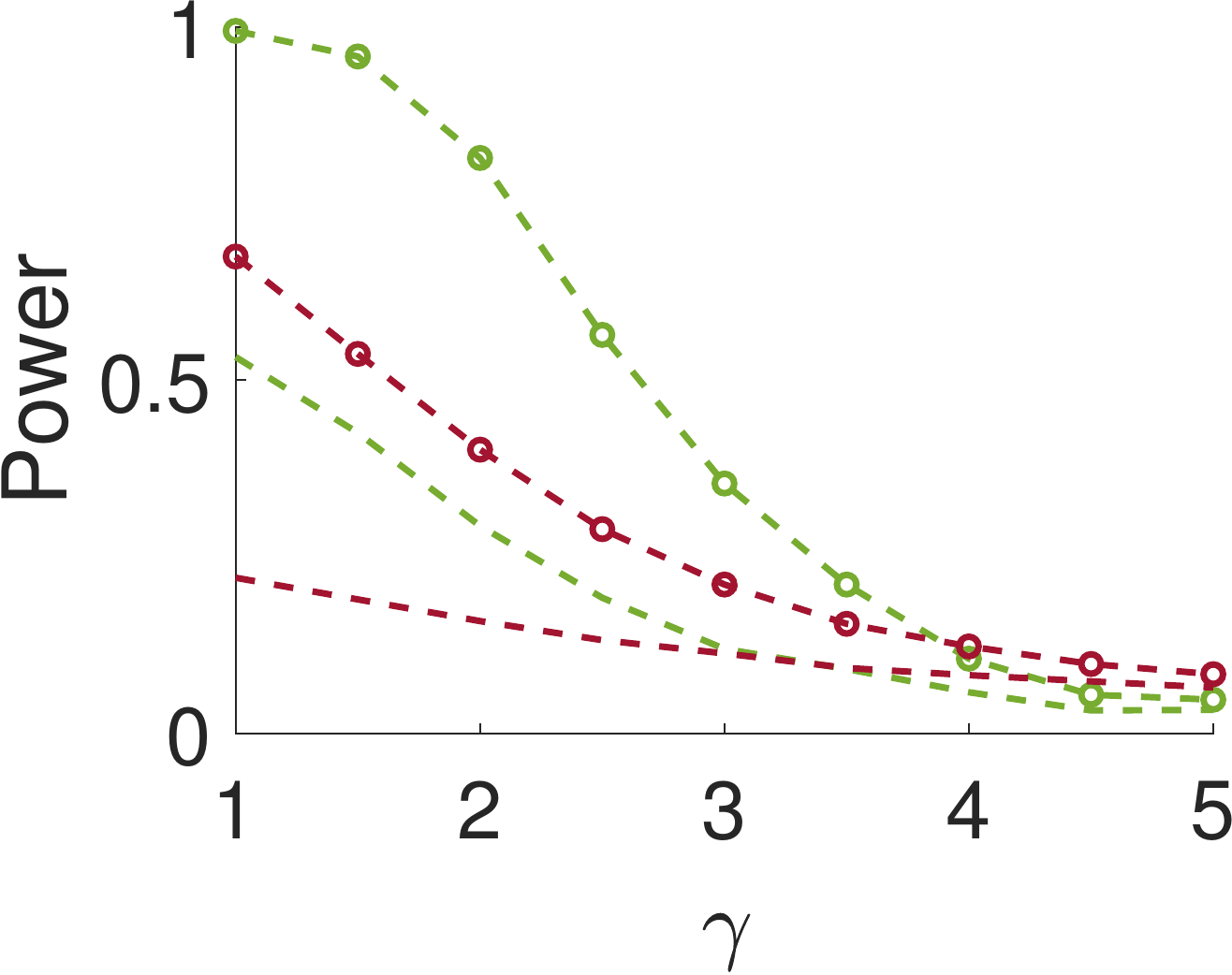}
\caption{\label{fig:neig_2_Power_nu_4} $\nu = 4$}
\end{subfigure}
\hfill
\begin{subfigure}[ht]{0.15\textwidth}
\centering
\includegraphics[width=\columnwidth]{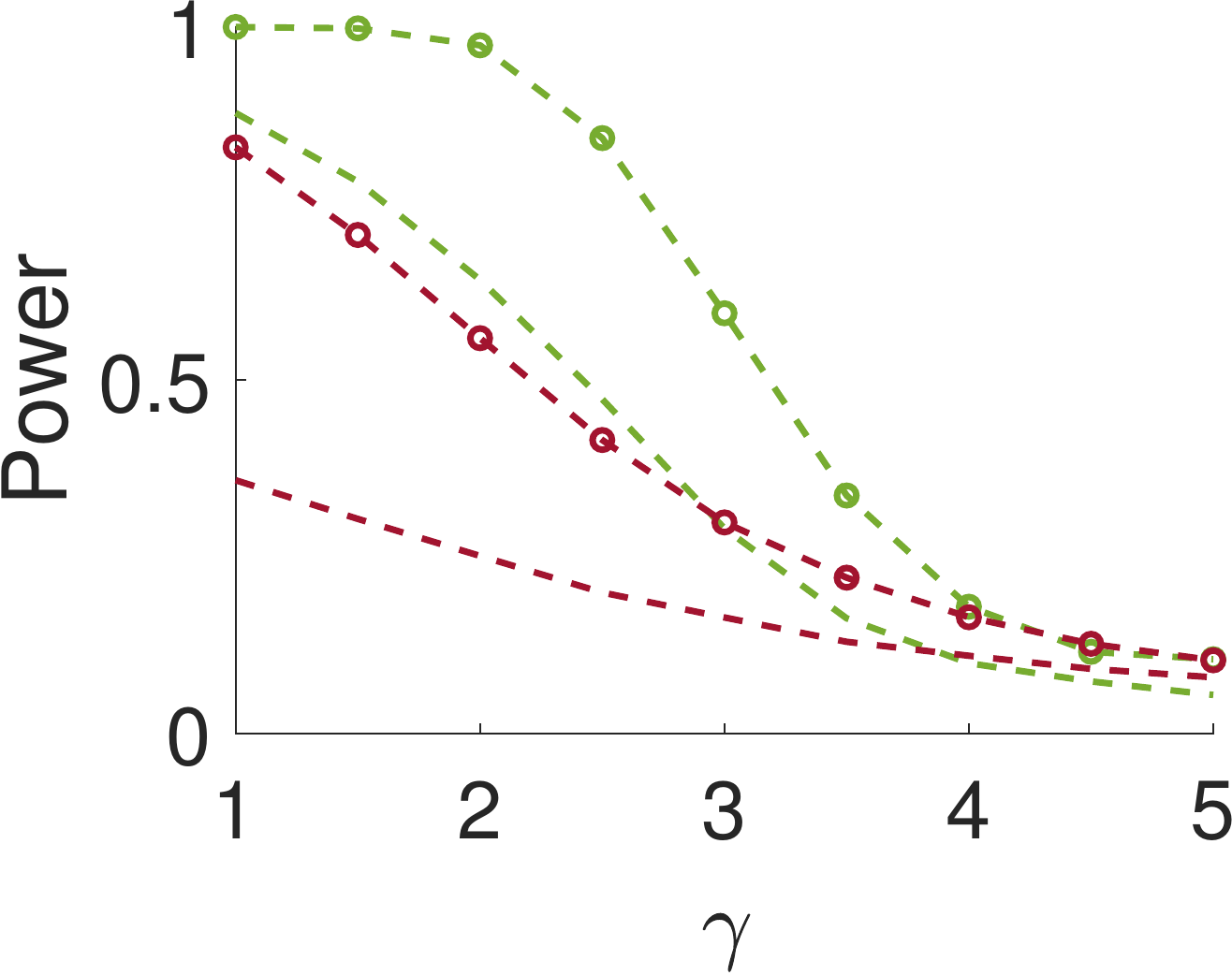}
\caption{\label{fig:neig_2_Power_nu_5} $\nu = 5$}
\end{subfigure}
\caption{\label{fig:dif b's for d =2}
Experiment with a neighbor at distance of two samples $d =  2$ and $a = 5$, where $a$ is the amplitude of the signal.
The red and green curves present, respectively, Algorithm~\ref{alg:one sample} and  Algorithm~\ref{alg:two sample}. We use '--' and '-o-', respectively, for $b = 3$ and $b = 2$, where $b$ is a measure for the signal's width~\eqref{eq:signal}.
}
\vskip -15pt
\end{figure}

Figure \ref{fig:dif b's for d =2} shows the power and FDR of Algorithms~\ref{alg:one sample} and~\ref{alg:two sample}, for a range of values  for  the parameters $\nu, \gamma, b$ and with a neighbor at distance of $d=2$ (for Algorithm~\ref{alg:two sample}).
While both tests control the FDR below $\alpha=0.05$  almost everywhere,
the power of the two-sample test is clearly greater than the power of the one-sample test.
The gap increases for smaller values of $b$, namely, for narrower signals. 
This phenomenon can be explained as for wider signals   the adjacent samples are similar to the local maxima; in the extreme, it can be thought of as sampling twice the same local maximum.

We repeat the same experiment with $d=5$, and the results are presented in Figure~\ref{fig:dif b's for d =5}.
As can be seen, the results for large $\nu$ are quite similar to the case of $d=2$. However, for small values of $\nu$, when 
the noise is weakly correlated, the power is clearly inferior.  This happens since, in this case, a sample $d=5$ entries away from the local maximum does not provide much information on the local maximum itself.

\begin{figure}
\begin{subfigure}[ht]{0.15\textwidth}
\centering
\includegraphics[width=\columnwidth]{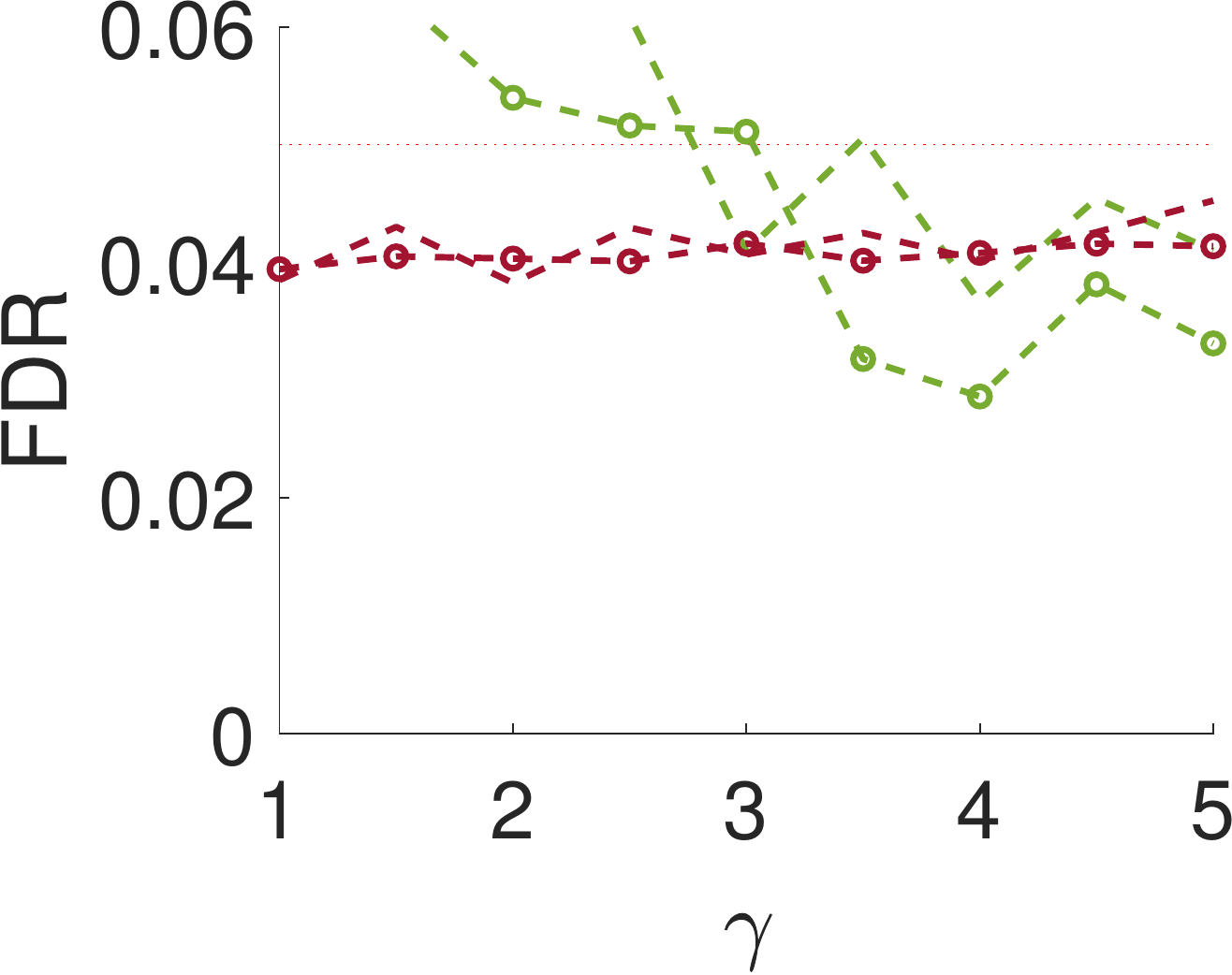}
\caption{\label{fig:neig_5_FA_nu_3} $\nu = 3$}
\end{subfigure}
\begin{subfigure}[ht]{0.15\textwidth}
\centering
\includegraphics[width=\columnwidth]{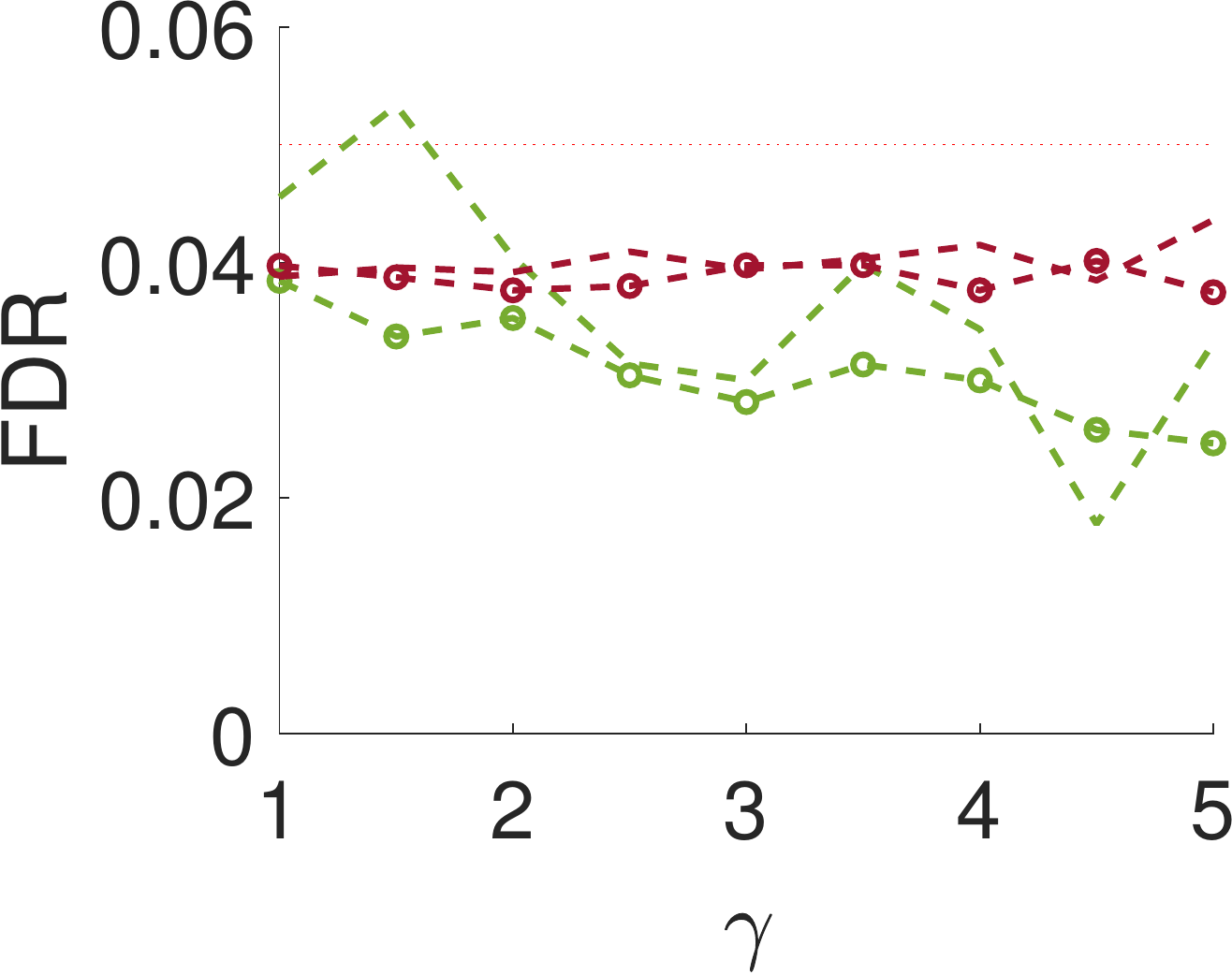}
\caption{\label{fig:neig_5_FA_nu_4} $\nu = 4$}
\end{subfigure}
\begin{subfigure}[ht]{0.15\textwidth}
\centering
\includegraphics[width=\columnwidth]{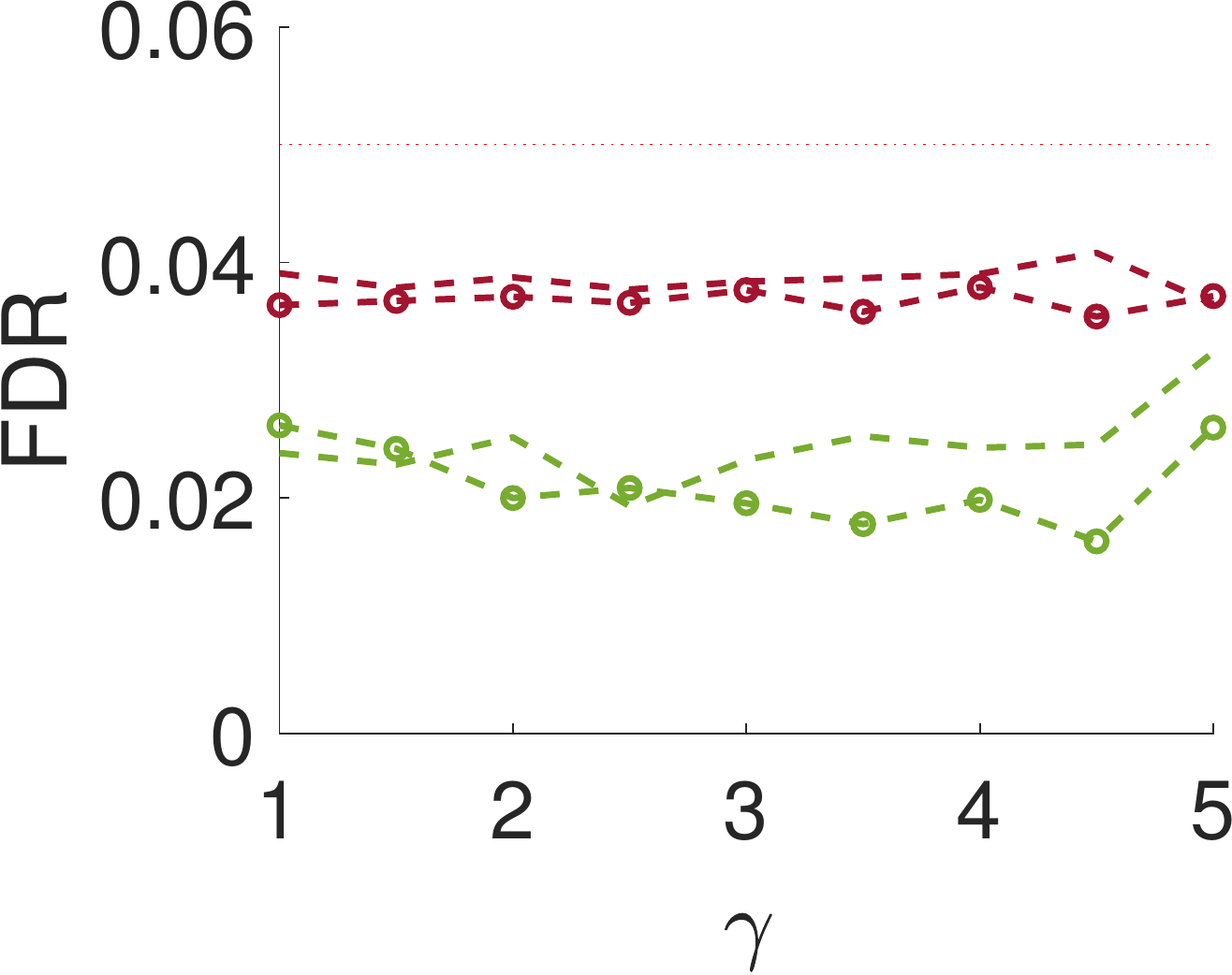}
\caption{\label{fig:neig_5_FA_nu_5} $\nu = 5$}
\end{subfigure}

\begin{subfigure}[ht]{0.15\textwidth}
\centering
\includegraphics[width=\columnwidth]{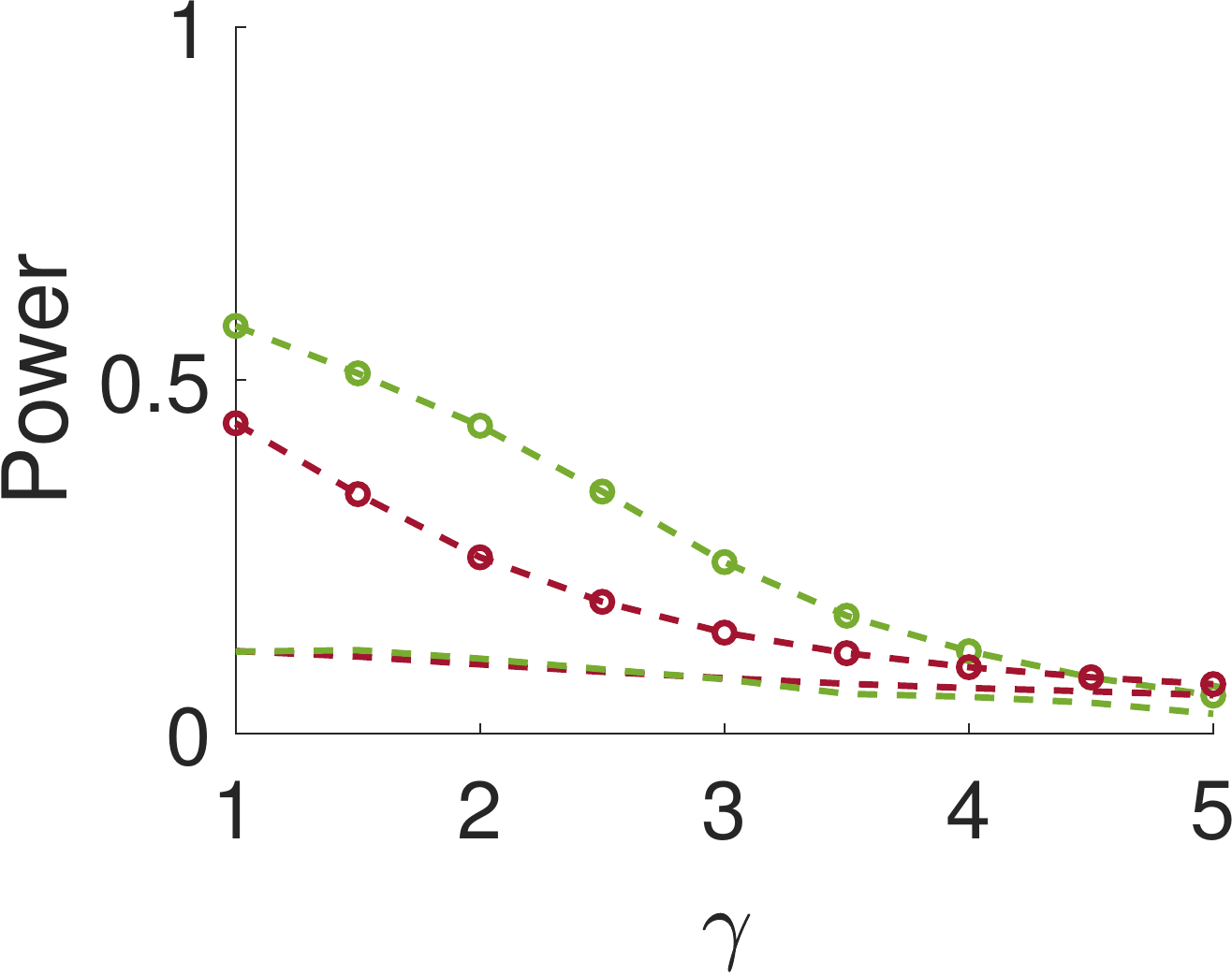}
\caption{\label{fig:neig_5_Power_nu_3} $\nu = 3$}
\end{subfigure}
\begin{subfigure}[ht]{0.15\textwidth}
\centering
\includegraphics[width=\columnwidth]{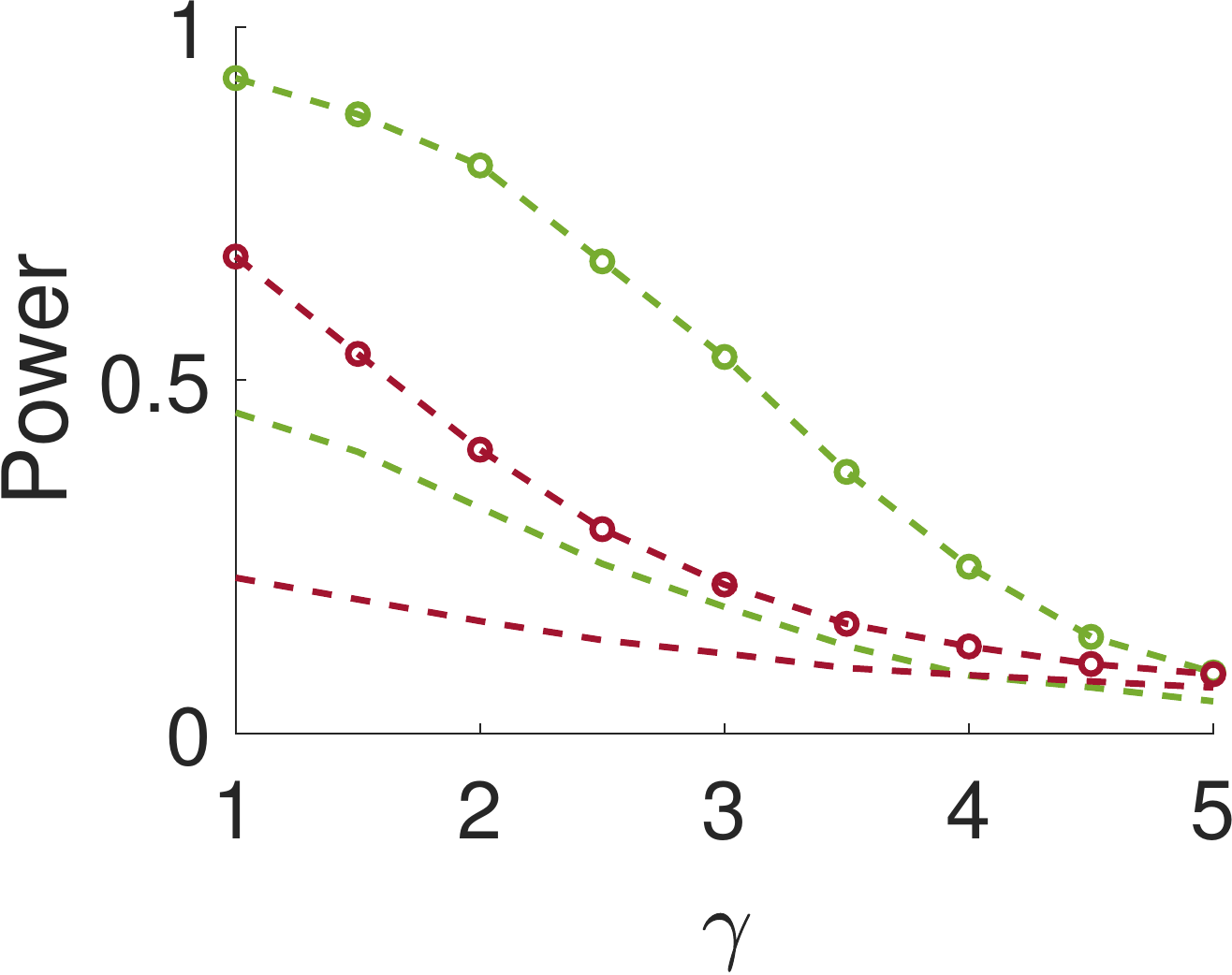}
\caption{\label{fig:neig_5_Power_nu_4} $\nu =4$}
\end{subfigure}
\begin{subfigure}[ht]{0.15\textwidth}
\centering
\includegraphics[width=\columnwidth]{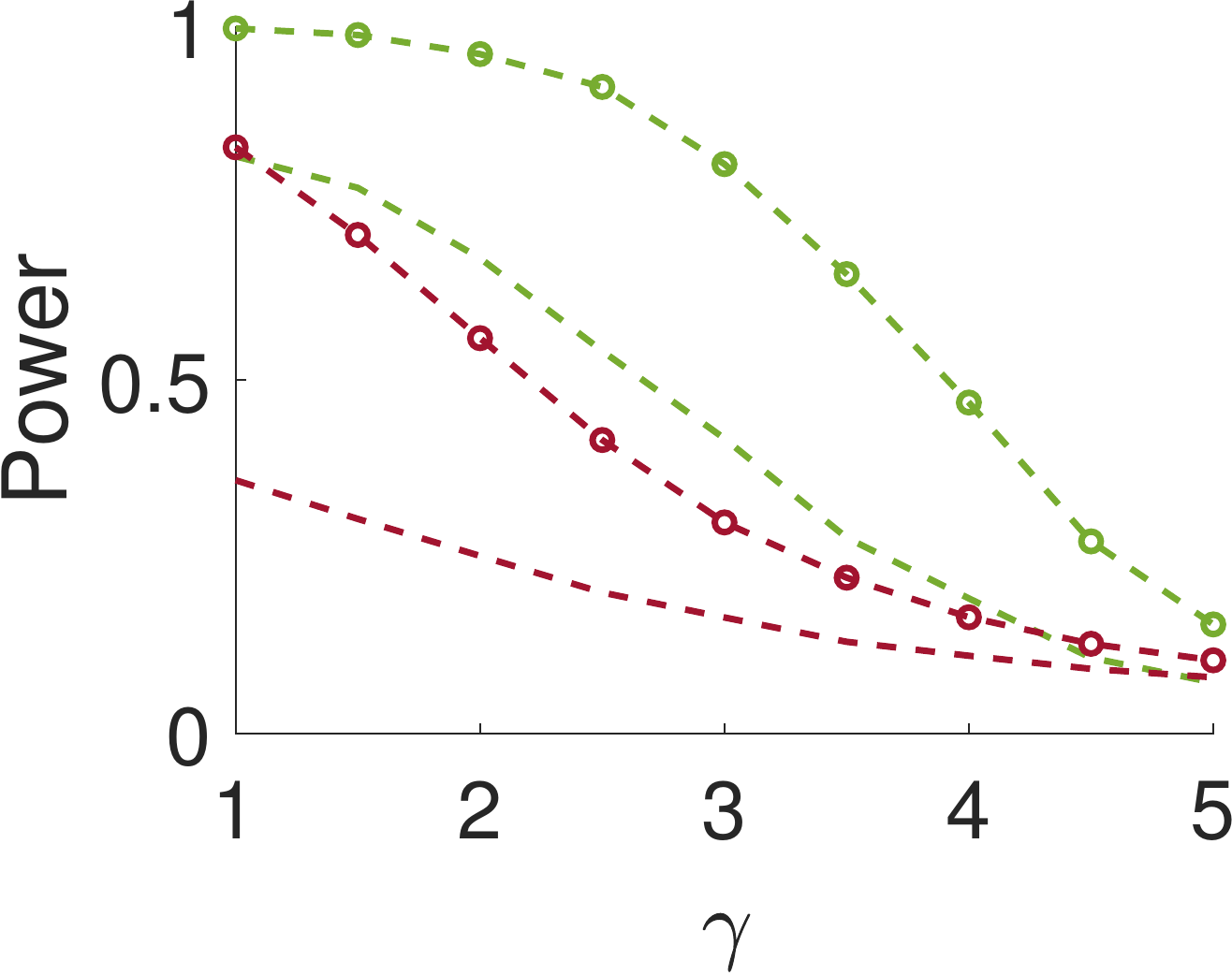}
\caption{\label{fig:neig_5_Power_nu_5} $\nu = 5$}
\end{subfigure}
\caption{\label{fig:dif b's for d =5}	
Experiment with a neighbor at distance of two samples $d =  5$ and $a = 5$, where $a$ is the amplitude of the signal.
The red and green curves present, respectively, Algorithm~\ref{alg:one sample} and  Algorithm~\ref{alg:two sample}. We use '--' and '-o-', respectively, for $b = 3$ and $b = 2$, where $b$ is a measure for the signal's width~\eqref{eq:signal}.
}
\end{figure}

Figure~\ref{fig:dif d's} compares Algorithm~\ref{alg:two sample} for  $d=3$ and $d=9$.
For all values of $\nu$, the case of $d=3$ outperforms $d=9$ for small values of $\gamma$ (the parameter of the smoothing kernel). For large $\gamma$, on the other hand, the latter outperforms the former. 
This happens since large $\gamma$  implies smoother noise and thus examining samples farther away from the local maxima provides more information.

\begin{figure}
\begin{subfigure}[ht]{0.15\textwidth}
\centering
\includegraphics[width=\columnwidth]{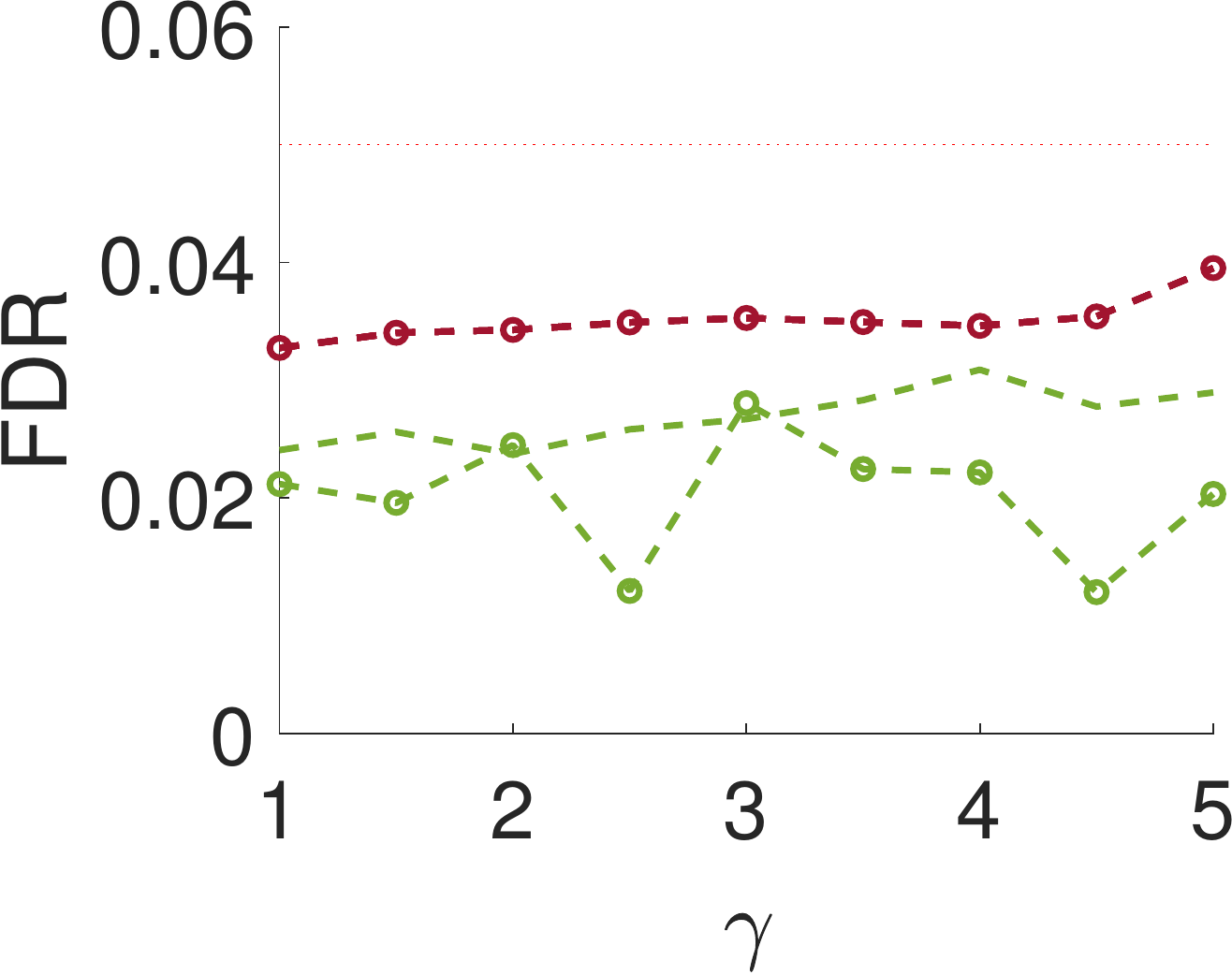}
\caption{\label{fig:dis3n9_FA_nu_6} $\nu = 6$}
\end{subfigure}
\begin{subfigure}[ht]{0.15\textwidth}
\centering
\includegraphics[width=\columnwidth]{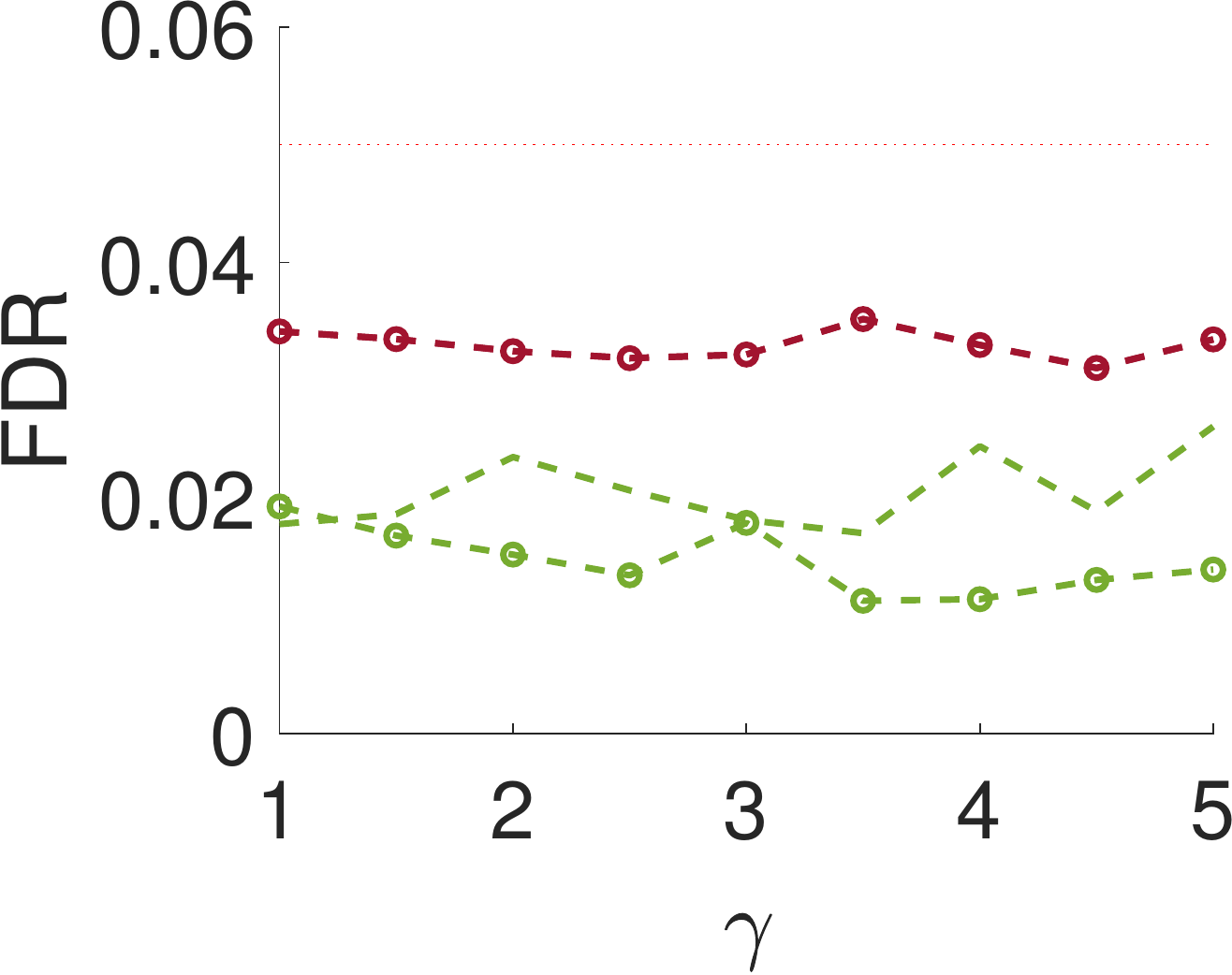}
\caption{\label{fig:dis3n9_FA_nu_7} $\nu = 7$}
\end{subfigure}
\begin{subfigure}[ht]{0.15\textwidth}
\centering
\includegraphics[width=\columnwidth]{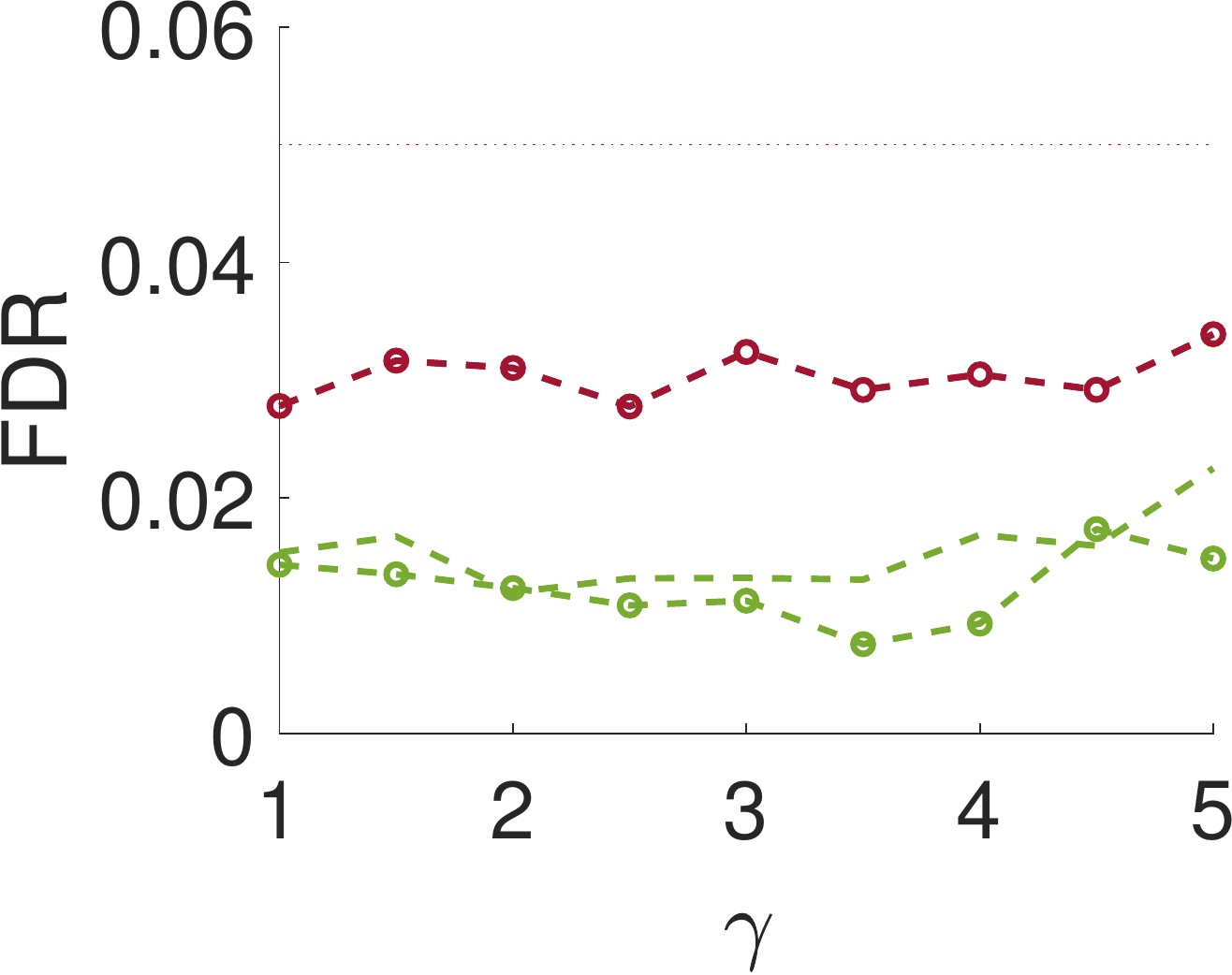}
\caption{\label{fig:dis3n9_FA_nu_8} $\nu = 8$}
\end{subfigure}

\begin{subfigure}[ht]{0.15\textwidth}
\centering
\includegraphics[width=\columnwidth]{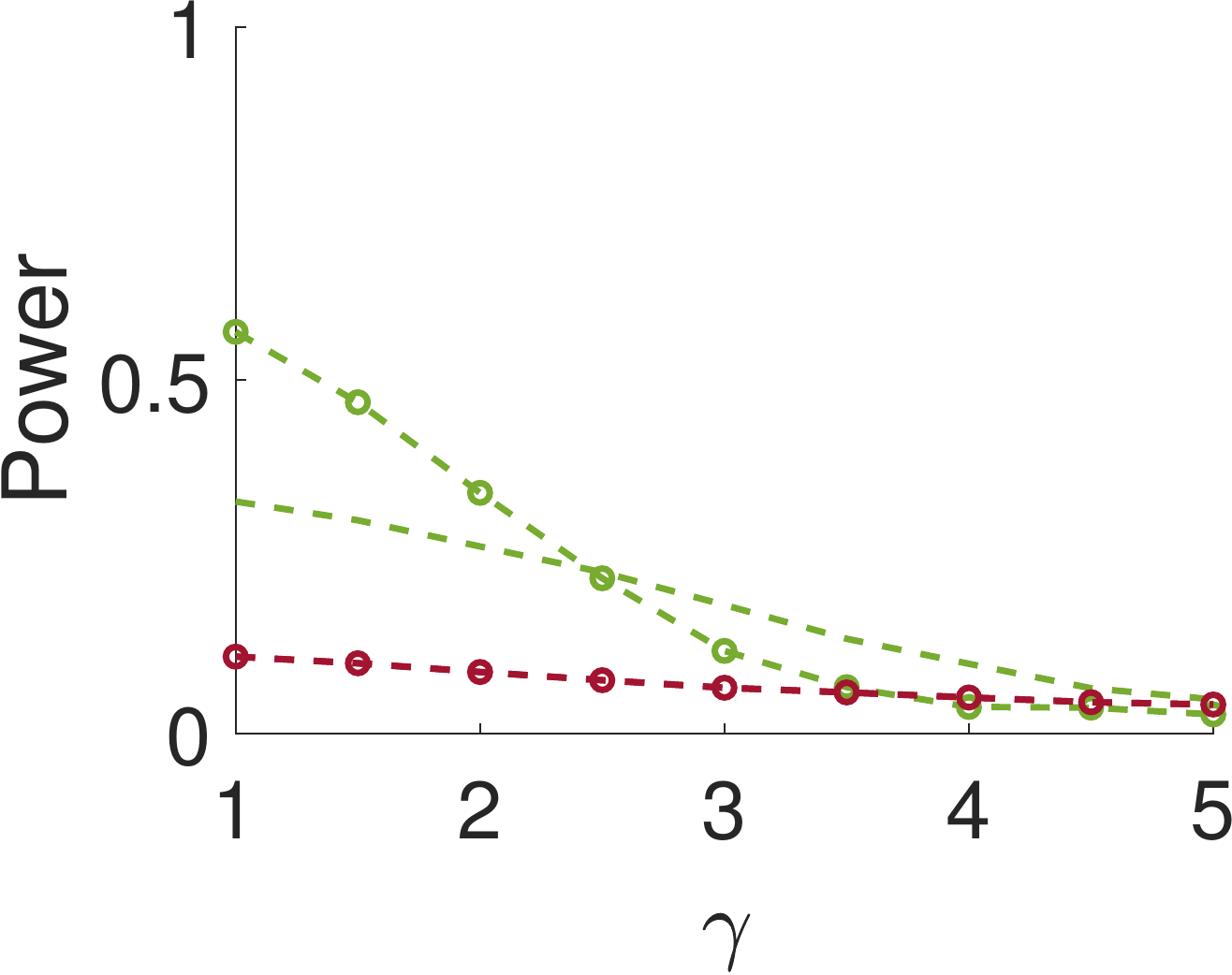}
\caption{\label{fig:dis3n9_Power_nu_6} $\nu = 6$}
\end{subfigure}
\begin{subfigure}[ht]{0.15\textwidth}
\centering
\includegraphics[width=\columnwidth]{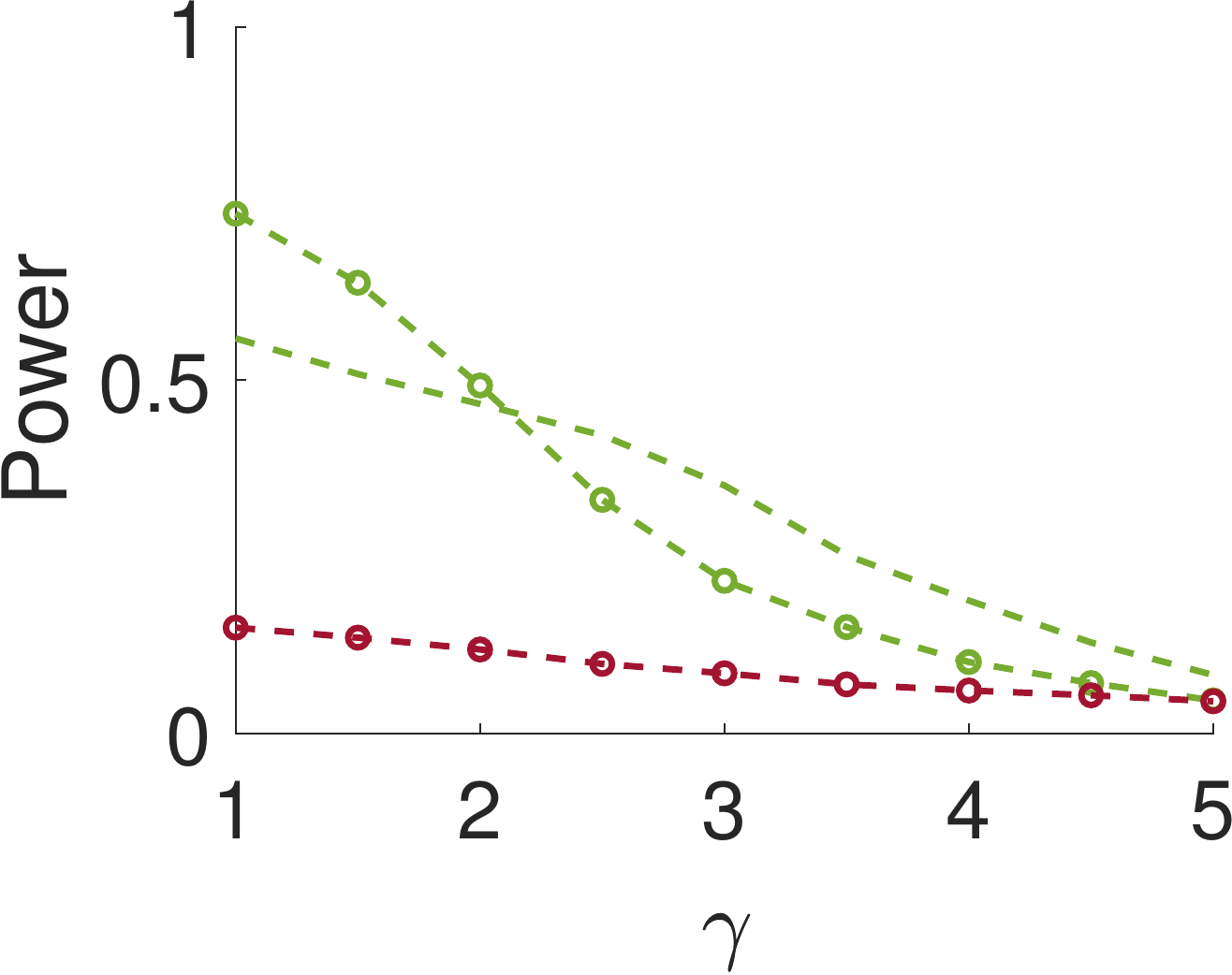}
\caption{\label{fig:dis3n9_Power_nu_7} $\nu =7$}
\end{subfigure}
\begin{subfigure}[ht]{0.15\textwidth}
\centering
\includegraphics[width=\columnwidth]{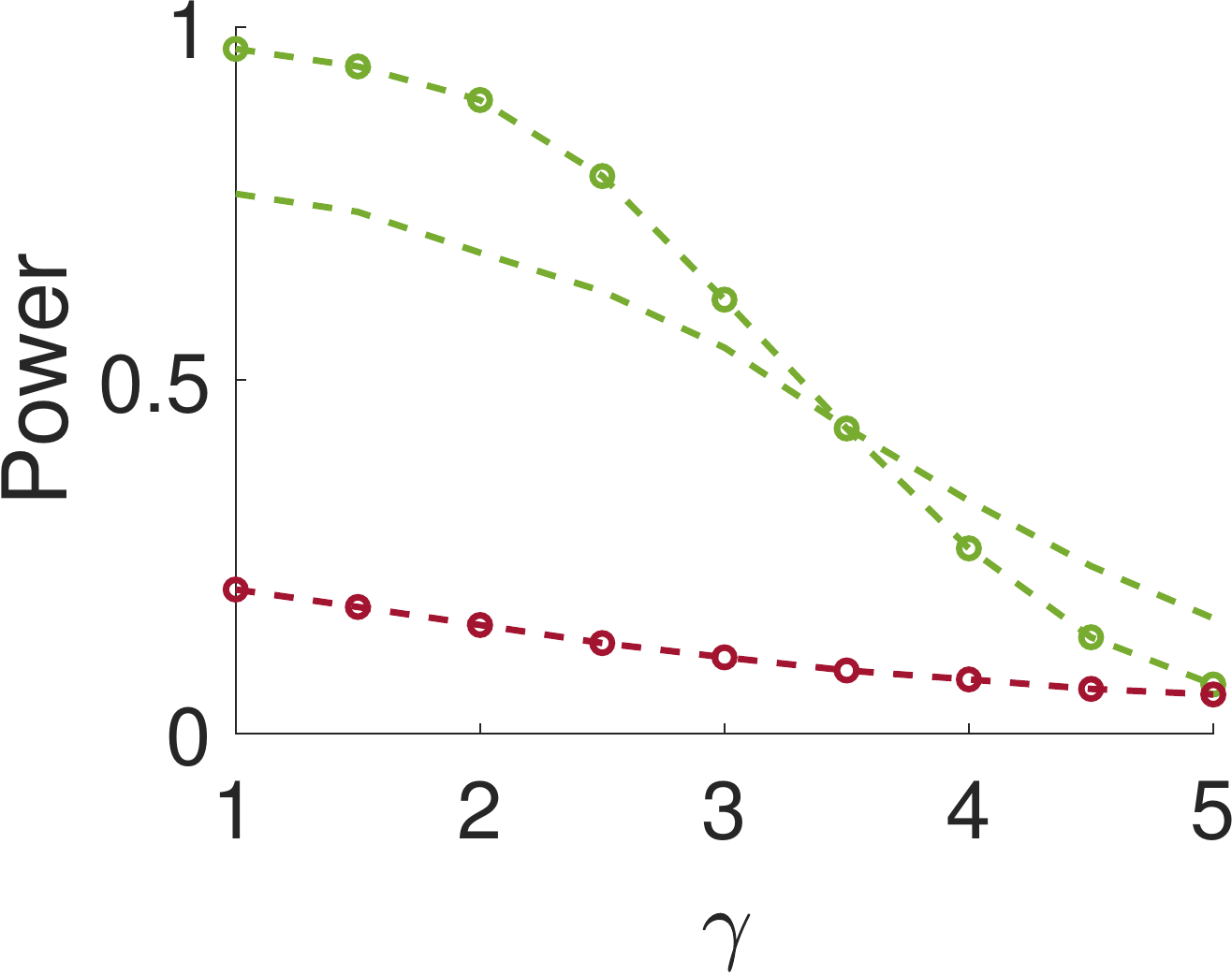}
\caption{\label{fig:dis3n9_Power_nu_8} $\nu =8$}
\end{subfigure}
\caption{\label{fig:dif d's}	
The red and green curves present, respectively, Algorithm~\ref{alg:one sample} and  Algorithm~\ref{alg:two sample}.
We use '--' and '-o-', respectively, for  Algorithm~\ref{alg:two sample} with $d=3$ and $d = 9$.
We set $b=3$ and $a=4$.
}
\vskip -15pt
\end{figure}

\section{Conclusions and discusssion}\label{sec:conclusion}

This paper introduces a new statistical test for detecting signal occurrences in a noisy measurement that controls the FDR below a desired level. 
In particular, our procedure does not only consider the local maxima of the measurement, but also its  neighbors. 
We numerically show that our two-sample test outperforms the one-sample test of~\cite{schwartzman2011multiple}  
for correlated noise (large enough $\nu$).
In particular, it seems that  for~$\nu$ about twice the variance of the signal $b$ and neighboring distance of around $d\approx 5$, our algorithm yields good performance.
The FDR of our algorithm is usually lower than the FDR of the one-sample test, probably due to the conservative  p-value definition~\eqref{eq:p_val}.

The main impediment of the proposed procedure is evaluating the multi-dimensional integrals~\eqref{eq:generalCase}, limiting the current implementation to a two-sample test. 
We aim to extend this work to a $K$-sample test with $K>2$. 
To this end, we intend to examine different approximation schemes of the integrals, perhaps along the lines of the truncated Gaussian approximation we use in this work. 
We expect that such extension will significantly improve the performance of the algorithm.  
Another future work is to follow the blueprint of~\cite{cheng2017multiple} and extend our $K$-sample test to images. This might be particularly important for the particle picking task: an essential step in the computational pipeline of single-particle cryo-electron microscopy~\cite{bendory2020single}.


\bibliographystyle{plain}


\end{document}